\documentclass[sigconf, 10pt, nonacm]{acmart}
\usepackage{amsmath}
\usepackage{amsthm}
\usepackage{booktabs}
\usepackage{cleveref}
\usepackage{enumerate}
\usepackage{stfloats}

\usepackage{todonotes}

\usepackage{tikz}
\usetikzlibrary{calc,fit,shapes,shapes.geometric,arrows.meta,positioning}
\usepackage{pgfplots}
\usepackage{pgfplotstable}
\usepgfplotslibrary{statistics}
\pgfplotsset{compat=1.18}

\newtheorem{theorem}{Theorem}
\newtheorem{lemma}{Lemma}

\newtheorem{definition}{Definition}

\newtheorem{corollary}{Corollary}

\makeatletter
\let\c@table\c@figure
\makeatother

\title[Byzantine Fault-Tolerant Post-Quantum Distributed Quorum Signatures]{Byzantine Fault-Tolerant Post-Quantum\\Distributed Quorum Signatures}

\author{Quentin Kniep}
\affiliation{%
  \institution{Anza}
  \country{Switzerland}
}

\author{Jakub Sliwinski}
\affiliation{%
  \institution{Anza}
  \country{Switzerland}
}

\author{Roger Wattenhofer}
\affiliation{%
  \institution{Anza \& ETH Zurich}
  \country{Switzerland}
}

\begin{abstract}
Threshold, aggregate, and multi-signatures---which we collectively call
quorum signatures---certify that a quorum of
nodes endorsed a statement, with a certificate as small as a single
signature. No constant-size post-quantum quorum signature is
known: all candidates grow with the number of signers and are slow to
aggregate, making quorum signatures the hardest obstacle to migrating
byzantine fault-tolerant systems to post-quantum security.

In this paper, we sidestep this open cryptographic problem by changing how the protocol
communicates. We introduce \emph{Distributed Quorum Signature (DQS)}, 
a primitive built solely from ordinary digital signatures and
a Bracha-style approval broadcast. DQS turns certificates from network
messages into local events. Two event types divide the roles certificates
play: \emph{weak certificates} capture safety, \emph{strong certificates} capture
liveness.
In DQS every message is constant size, fitting a single datagram
regardless of the number of nodes. The total communication is
quadratic, and no security assumptions change.
In a large distributed system, the overhead of post-quantum DQS is competitive with the canonical pre-quantum BLS scheme.

\end{abstract}

\begin{document}

\maketitle

\section{Introduction}

\begin{table*}[b]
\centering
\begin{tabular}{@{}llllll@{}}
    \toprule
    Scheme & Messages & Quorums & Interaction & Attribution & Transferability \\
    \midrule
    Threshold sig. & same & thresholds only$^\dagger$ & setup required & not generally & succinct \\
    Aggregate sig. & arbitrary & arbitrary & none & direct from sig. & succinct sig. + $\mathcal{O}(n)$ bitmap \\
    DQS & arbitrary & arbitrary & all-to-all$^\ddagger$ & some correct node$^*$ & $\mathcal{O}(n)$ size \\
    \bottomrule
\end{tabular}
\caption{\label{tab:comparison}Comparison of approaches to quorum certification.
$^\dagger$Arbitrary weights must be virtualized.
$^\ddagger$Direct or indirect.
$^*$While not all nodes creating the quorum certificate can attribute the signers, we ensure that there is at least one correct node that holds the cryptographic proof to do so.}
\end{table*}


A powerful quantum computer would pose a grave threat to internet security.
Many established cryptographic primitives, such as digital signatures, could be attacked and broken by a quantum computer.
Quantum computing may not be practical yet, but protocols must be ready before it is too late.

Fortunately, many cryptographic primitives already have post-quantum alternatives. These alternatives, however, often come with substantial performance trade-offs:
Naively swapping a pre-quantum (quantum insecure) primitive for its post-quantum (quantum secure) counterpart can impose a severe cost.

The situation is particularly dire for the widely used advanced signature primitives, like aggregate signatures, threshold signatures, and \textit{multiple} multi-signature variants.
Each of these allow producing combined cryptographic signatures for a quorum of signers, yet remain as short as a signature from a single party.
We use \emph{quorum signature} as an umbrella term for such primitives.
See \Cref{tab:comparison} for a comparison. 

BLS aggregate signatures~\cite{agg_bls,BLS04} allow ad hoc aggregation of signatures over the same message by simply adding up individual signatures, producing an aggregate of just 48 to 192 bytes, depending on the parameterization.
This compactness has made BLS aggregate signatures a popular ingredient in many pre-quantum distributed systems.
As a result, when trying to migrate a distributed system to post-quantum security, quorum signatures are often the hardest piece to replace.

While post-quantum quorum signatures exist, there is no post-quantum alternative with constant size aggregates.
The existing proposals all grow with the number of signers and involve slow aggregation.
In the most compact proposals even verification is unacceptably slow.
Because of this, quorum signatures are a true show-stopper when it comes to making distributed protocols secure against potential quantum attacks.

To the best of our knowledge, no roadmap exists for achieving efficient
post-quantum quorum signatures of any kind.
In this paper, we present a radically different approach: rather than tackling this open cryptographic problem head-on, we sidestep aggregation entirely.

Our approach changes the semantics.
In established (non-distributed) quorum signatures, one party collects a quorum of individual signatures and locally combines them into a quorum signature that serves as a certificate that the quorum was achieved.
We instead introduce a \textit{Distributed Quorum Signature (DQS)}:
Rather than being computed locally, the DQS is the result of a distributed communication protocol.
Instead of a single concrete certificate, the protocol produces two distinct types of certificates, each carrying its own novel semantics, both useful in distributed protocol design.

We define DQS in \Cref{def:interactive-certificate-creation} with the notions of weak and strong certificates.
Certificates correspond to quorum signatures, but are local events resulting from messaging in the distributed system rather than self-contained proofs like quorum signatures.

For a weak certificate to be created by any correct node, some correct node must have observed the necessary quorum of votes.
In other words, it is as \textit{safe} for a protocol to act based on a weak certificate, as it is based on a quorum signature.
However, if some of these votes were cast by byzantine nodes, it might be that other correct nodes will not create the same weak certificate.

A strong certificate features an additional property: If it is observed by any correct node, all other correct nodes will also observe the certificate.
Moreover, if enough correct nodes cast votes to meet the quorum, the corresponding strong certificate will necessarily be created.
In other words, the malicious nodes are never needed for the protocol to progress with strong certificates.
Since correct nodes are enough to produce strong certificates, and strong certificates are observed by all correct nodes, \textit{liveness} can rely on strong certificates.

In executions of the protocol without misbehaving nodes, both weak and strong certificates will be observed by all nodes whenever the needed quorum is met.
With misbehavior, some nodes might create weak certificates without creating the same strong certificate.
However, we ensure that if a strong certificate is observed, all correct nodes also observe the weak certificate.

DQS is instantiated with a regular post-quantum signature scheme, so can be used without requiring subtle cryptographic primitives.
DQS adds only a small constant overhead to each message, such that for small enough signatures and small enough vote payloads all messages can fit in an MTU-sized datagram, regardless of the size of the distributed system.
In summary, post-quantum security does not have to mean a slower, clunkier protocol.

\section{Model}
\label{sec:model}

\paragraph{Node.}
We have a distributed system with $n$ individual computers, which we call nodes 
$v_1,v_2,$ $\ldots,$ $v_n$.
We assume that the set of nodes is fixed and publicly known, i.e., each node knows how to contact (IP address and port number) every other node.
Each node has a public key, and all nodes know all public keys.

\paragraph{Weight.}
Nodes have different voting weights. 
Each node $v_i$ has a known weight $\rho_i > 0$
to denote node $v_i$'s fraction of the entire weight, i.e., $\sum_{i=1}^n \rho_i = 1$.
In the symmetric case, every node has the same weight, i.e., $\rho_i = 1/n$.

\paragraph{Message.}
Nodes communicate by exchanging authenticated messages over the internet.
Our protocol never uses large messages.
Specifically, every message 
fits into a single MTU datagram~\cite{ethernet_rfc} (depending on the signature scheme and assuming the vote payloads are small enough, roughly $< 400$~B).
Because of this, we can use UDP with authentication, e.g., QUIC-UDP.

\paragraph{Broadcast.}
Sometimes, a node needs to broadcast the same message to \textit{all} ($n-1$ other) nodes.
The sender node simply loops over all other nodes and sends the message to one node after the other. If available, we could also use a multicast service.

\paragraph{Adversary.}
Some nodes can be byzantine in the sense that they can misbehave in arbitrary ways.
Byzantine nodes can for instance forget to send a message.
They can also collude to attack the system in a coordinated way.
We assume that all the byzantine nodes together own up to $f$ of the total weight.
Additionally, nodes with weight up to $c$ may crash at any time.
The remaining nodes with weight at least $1-f-c$ are \textit{correct} and follow the protocol.

\paragraph{Fault Tolerance.}
We assume that $3f + 2c < 1$ for our construction to be correct. Popular parameter choices for this assumption are $f<1/5,c=1/5$ or $f<1/3, c=0$.

\paragraph{Asynchrony.}
We consider the partially synchronous network setting of Global Stabilization Time (GST) \cite{PartialSynchrony,unifying_partial}.
Messages sent between correct nodes will eventually arrive, but they may take arbitrarily long to arrive.

\paragraph{Synchrony.}
In the model of GST, synchrony simply corresponds to a global worst-case bound $\Delta$ on message delivery.
The GST model captures periods of synchrony and asynchrony by stating that before the unknown and arbitrary time $GST$ messages can be arbitrarily delayed, but after time GST all previous and future messages $m$ sent at time $t$ will arrive at the recipient at latest at time 
$\max(\textsf{GST}, t)+\Delta$. 




\paragraph{Votes and Quorum Signatures}
Our construction provides an alternative for protocols making use of votes and quorum signatures defined below.
In this general setup, nodes vote on a specific event. Nodes receive and remember the votes of other nodes. If some node sees enough votes for an event, they can generate a quorum signature.

\begin{definition}[vote]\label{def:vote}
Given a set of messages.
A vote is a message signed by a single node.
\end{definition}

\begin{definition}[quorum signature]\label{def:certificate}
Consider a predicate $P$ over sets of votes, such that if $P(V)$ is true for some $V$, then for any $V' \supseteq V$, $P(V')$ is also true.
Then, a quorum signature is a set of signatures (individual, aggregate, threshold or combinatorial) proving that $P(V)$ is true and votes $V$ were cast by nodes.
\end{definition}

\begin{table*}[b]
\begin{tabular}{@{}lcccccc@{}}

    \toprule
    Scheme & Type & Sig. & PK & Cat. & Sign.\ Risk & Verify \\
    \midrule
    \multicolumn{6}{@{}l}{\emph{Standardized or in active standardization}}\\
    Falcon-512~\cite{falcon}            & Lattice & 666\,B      & 897\,B  & 1 & High     & Fast     \\
    XMSS-SHA2\_20\_192~\cite{xmss_nist} & Hash    & 1.7~KB & 52~B   & 3 & Moderate & Fast \\
    ML-DSA-44~\cite{nist_ml_dsa}        & Lattice & 2.4~KB & 1.3~KB & 2 & Low      & Fast \\
    SLH-DSA-128~\cite{nist_slh_dsa}     & Hash    & 7.9~KB & 32~B   & 1 & Low      & Slow \\
    \midrule
    \multicolumn{6}{@{}l}{\emph{Non-standardized (research)}}\\
    SQISign~\cite{sqisign}              & Isogeny & 148~B  & 65~B   & 1 & High     & Slow \\
    Hawk-512~\cite{hawk}                & Lattice & 555~B  & 1.0~KB & 1 & Low      & Fast \\
    \bottomrule
\end{tabular}
\caption{\label{tab:pq-signatures}Candidates for post-quantum signature schemes with small signatures.
Sig./PK: signature and public-key sizes.
Cat.: NIST security category (1, 2, 3, 4, 5; higher is stronger).
Sign.\ Risk: risk of a signing-side failure (secret-state
reuse for XMSS, side-channel leakage for
Falcon/SQISign).
Verify: relative verification speed.
}
\end{table*}

\section{Related Work}
\label{sec:primitives}




In this section, we discuss several quantum-proof cryptographic tools.
Notably, many fundamental cryptographic primitives remain effective even in the presence of practical quantum computers.
In particular, symmetric encryption, message authentication mechanisms, and cryptographic hash functions are considered to remain secure against quantum attacks.

While quantum algorithms such as Grover's algorithm can provide a quadratic speedup for brute-force search, they do not fundamentally break these primitives, and the resulting loss in security can generally be compensated by increasing key and digest sizes \cite{grover1996fast,bernstein2009pqc,grassl2016applying}.
Moreover, many hash-based constructions and security proofs have been studied explicitly in the quantum setting \cite{boneh2011random}.

Unfortunately, the popular elliptic curve signature schemes do not carry over to the world of quantum computing.
We want to specifically consider the use case where each message of a protocol should fit into a single MTU-sized network datagram.
Assuming QUIC is used in datagram mode for transport, 1{,}160\,B remain available for payload.

Table \ref{tab:pq-signatures} summarizes the state of the art of post-quantum signature schemes. 
The smallest to be standardized post-quantum signature is Falcon-512~\cite{falcon}.
All other standardized signatures do not fit a datagram MTU size.
However, there are some interesting alternatives that are not standardized (yet):
Hawk-512 \cite{hawk} is closely related to Falcon-512; it eliminates the floating-point sampling hazard in the signing procedure, which can leak the secret key if implemented without constant-time guarantees.


While XMSS-SHA2\_20\_192~\cite{xmss_nist} signatures are too large, XMSS-like~\cite{xmss} hash-based constructions can be parametrized.
In \Cref{sec:custom_xmss} we discuss several parameter options and sketch how even using one-time signatures directly can be feasible.
In summary, while signatures incur an overhead that demands a redesign of a protocol, there are feasible solutions.

\begin{table*}[t]
\begin{center}
\begin{tabular}{@{}lccccccc@{}}
    \toprule
    Scheme & Type & Messages & Single Sig. & Agg.\ Sig. & Growth & Aggregate & Verify \\
    \midrule
    naive Falcon-512 & Lattice & arbitrary & 666\,B & $\le$ 666\,KB & $\Theta(n)$ & $\approx$ 0 & $\le$ 40\,ms \\
    \midrule
    Lemur~\cite{lemur} & Lattice & same, sync. & $\approx$ 78\,KB & $\approx$ 185\,KB & $\mathcal{O}(\log n)$ & $\approx$ 1\,s & $\approx$ 15\,ms \\ 
    leanMultisig~\cite{lean_multisig,leanvm} & Hash & same & $\approx$ 1.5\,KB & $\approx$ 400\,KB & polylog & $\approx$ 2--3\,s & $\approx$ 30\,ms \\
    Falcon512+LaZer~\cite{lazer} & Lattice & arbitrary & 666\,B & $\approx$ 70\,KB & polylog & $\approx$ 500\,ms & $\approx$ 250\,ms \\ 
    \midrule
    BLS~\cite{agg_bls} & Pairing & arbitrary & 48--192\,B & 48--192\,B & $\Theta(1)$ & $\approx$ 1\,ms & $\approx$ 0.3--0.5\,ms \\
    \bottomrule
\end{tabular}
\caption{\label{tab:pq-agg-signatures}Post-quantum aggregatable signature candidates at $n \approx 1{,}000$ signers; pre-quantum BLS for reference.
Synchronized schemes (sync.)\ sign the same message, and only if signers signed them for the same time step.
Timings as reported by the respective works, on differing hardware.
All schemes additionally need ${\approx}\,n$ bits to identify the signer set.
Naive Falcon-512 concatenates: aggregation is free, but size and verification grow with $n$.}
\end{center}
\end{table*}

What's even more difficult to replace are BLS aggregate signatures, which are commonly used in protocols to allow aggregation of individual nodes' votes into certificates for a quorum of signers.
While BLS signatures are elegant, they usually require a bitmap indicating which individual signatures are included.
If we want a BLS aggregate signature to fit within a single datagram, this bitmap constrains the system to roughly $n \approx 10{,}000$ nodes.
In our approach, by contrast, all messages have strictly constant size and fit comfortably inside a datagram regardless of $n$.

Currently no construction is known for constant-size quorum signatures with post-quantum security.
The dedicated synchronized multi-signature constructions that exist~\cite{squirrel,chipmunk,lemur} are logarithmic in the number of signers with large constant factors, making them only slightly more efficient than naive concatenation at $n \approx 1{,}000$ nodes.
Other, more general, attempts to achieve aggregation are based on zero-knowledge proofs over the verification of all the individual signatures, which are in the same order of magnitude regarding size, and with prover latency in the order of seconds.
See \Cref{tab:pq-agg-signatures} for a comparison of some of the most practically efficient aggregate signature proposals.

In summary, the situation for aggregate signatures appears bleak.
Accordingly, in this paper, we completely eliminate the need for signature aggregation by introducing a new communication pattern for a (consensus) protocol.

\section{Distributed Quorum Signature}
\label{sec:dqs}

We are refraining from using any cryptographic primitive other than simple signatures as in \Cref{def:vote}, and want to produce an interactive protocol providing functionality of quorum signatures of \Cref{def:certificate}.
In our protocol, we introduce two different notions of local events approximating propagation of quorum signatures between nodes. We call the events weak and strong certificate, or cert for short. \Cref{def:interactive-certificate-creation} states the properties we aim for.

\begin{definition}[distributed quorum signatures]
\label{def:interactive-certificate-creation}
    A distributed quorum signature scheme is an interactive protocol,
    which starts with the nodes issuing votes at potentially different times and ensures the following properties:
\begin{itemize}
    \item \textsc{(weak cert safety)} If any correct node creates a weak certificate, then enough nodes voted to meet the quorum.
    \item \textsc{(strong cert liveness)} If enough correct nodes vote to meet the certificate creation threshold, then all correct nodes will eventually create the strong certificate.
    \item \textsc{(consistency)} If any correct node creates a strong certificate, then all correct nodes will create the strong and the corresponding weak certificate.
\end{itemize}
    
\end{definition}

\subsection{Protocol}

In this section we describe our DQS protocol.
DQS uses regular digital signatures and Bracha-style reliable broadcast \cite{bracha1987asynchronous} as building blocks.

\begin{figure*}
\begin{tikzpicture}[
    >={Stealth[round]},
    box/.style={draw, align=center, inner sep=5pt, font=\small},
    font=\small
]
  \node[box] (send) {broadcast approval\\\textit{and}\\create weak certificate};
  \node[box, right=40mm of send] (pool) {create\\strong certificate};
  \draw[->] ($(send.west)+(-30mm,4mm)$) -- ($(send.west)+(0mm,4mm)$)
        node[midway, above, align=center]
        {sufficient votes};
  \draw[->] ($(send.west)+(-30mm,-4mm)$) -- ($(send.west)+(0mm,-4mm)$)
        node[midway, below, align=center]
        {$> f$ approvals};
  \node[font=\itshape] at ($(send.west)+(-15mm,0mm)$) {or};

  \draw[->] (send.east) -- (pool.west)
        node[midway, below, align=center]
        {$> 2f + c$ approvals};

\end{tikzpicture}
\caption{Certificate creation and approval broadcast.}
\label{fig:approval-flow}
\end{figure*}
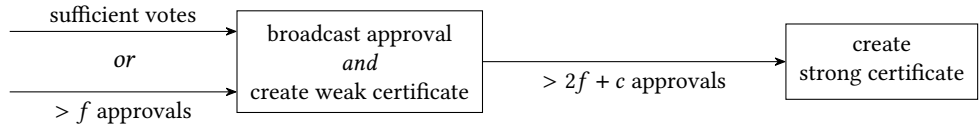

\begin{definition}[vote broadcast]\label{def:vote-broadcast}
Votes of correct nodes are broadcast to all nodes.
\end{definition}

\begin{definition}[approval]
\label{def:approval}
    Approval is a message containing a certificate payload, signed by a single node.
\end{definition}

\begin{definition}[weak cert and approval broadcast]
\label{def:weak-cert-creation}
A weak certificate is created when one of the following conditions is met for the first time:
\begin{itemize}
    \item Votes matching the corresponding certificate creation quorum are observed.
    \item The corresponding approvals from nodes with cumulative weight of more than $f$ are observed.
\end{itemize}
Moreover, when a weak certificate is created the corresponding approval is broadcast.
\end{definition}

\begin{definition}[strong cert]\label{def:strong-cert-creation}
A strong certificate is created when a node observes $>2f+c$ corresponding approvals.
\end{definition}


\subsection{Properties}

\begin{lemma}[weak cert safety]
\label{lem:weak-cert-safety}
    Our protocol satisfies weak cert safety of \Cref{def:interactive-certificate-creation}.
    Moreover, if some correct node creates a weak certificate, there is some (potentially different) correct node that collected enough votes for creating the certificate.
\end{lemma}

\begin{proof}
    Let $v$ be the correct node that created this weak cert first in the execution.
    By \Cref{def:weak-cert-creation}, one of the following conditions holds:
    \begin{enumerate}[a)]
    \item The node $v$ observed votes matching the certificate creation quorum itself. Then, trivially, enough nodes voted for the certificate to exist.
    \item The node $v$ created the weak certificate after receiving the corresponding approval from nodes with $>f$ of weight.
    Since byzantine nodes hold at most $f$ of weight, some correct node broadcast the corresponding approval before $v$.
    However, the approval is only broadcast upon observing the weak certificate, contradicting the choice of $v$.\qedhere
    \end{enumerate}
\end{proof}

\begin{lemma}[strong cert liveness]
\label{lem:strong-cert-liveness}
    DQS satisfies strong cert liveness of \Cref{def:interactive-certificate-creation}.
\end{lemma}
\begin{proof}
    As per the property condition, suppose the votes of correct nodes meet a certificate condition. By \Cref{def:vote-broadcast}, all correct nodes will observe these votes. By \Cref{def:weak-cert-creation}, each correct node will cast a corresponding approval as a result. Since $1 > 3f + 2c$, each node will receive approvals of correct nodes with weight $1 - f - c > 2f + c$, fulfilling the condition of \Cref{def:strong-cert-creation}. Therefore each correct node will create the strong cert.
\end{proof}

\begin{lemma}[consistency]
\label{lem:consistency}
    DQS satisfies consistency of \Cref{def:interactive-certificate-creation}.
    Moreover, if a correct node creates a strong certificate at time $t$, all correct nodes create the same strong certificate by time $\max(t+2\Delta, \textit{GST}+2\Delta)$.
\end{lemma}
\begin{proof}
    Suppose a correct node $v$ created a strong certificate at time $t$.
    By \Cref{def:strong-cert-creation}, $v$ observed approval messages for the certificate from nodes with $> 2f + c$ of weight.
    Since nodes with weight of at most $f + c$ are faulty, correct nodes with weight $> 2f + c - (f+c) = f$ broadcast approvals for this certificate payload by time $t$.
    By definition of GST, these approvals will reach all correct nodes by time $\max(t+\Delta, \textit{GST}+\Delta)$.
    Since all correct nodes will receive approvals from correct nodes with weight $>f$, by \Cref{def:weak-cert-creation} all correct nodes will broadcast their approvals for this certificate payload by time $\max(t+\Delta, \textit{GST}+\Delta)$.
    Then, by time $\max(t+2\Delta, \textit{GST}+2\Delta)$, all correct nodes will receive approvals from all correct nodes with weight $1 - f - c > 2f+c$. By \Cref{def:strong-cert-creation}, all correct nodes will create the corresponding strong certificate by time $\max(t+2\Delta, \textit{GST}+2\Delta)$.
\end{proof}

\begin{theorem}
    Our protocol is a distributed quorum signature scheme of \Cref{def:interactive-certificate-creation}.
\end{theorem}
\begin{proof}
    The properties follow from \Cref{lem:weak-cert-safety,lem:strong-cert-liveness,lem:consistency}.
\end{proof}

\begin{lemma}[communication complexity]
\label{lem:communication-complexity}
Suppose the vote payloads are encodable in constant size.
Then, each message of DQS has size $\mathcal{O}(1)$, and it sends $\mathcal{O}(n^2)$ messages per weak certificate that is created.
\end{lemma}
\begin{proof}
    The size of vote and approval signatures is independent of the number of nodes.
    By \Cref{def:vote,def:approval}, all messages are then of size $\mathcal{O}(1)$.

    Each weak cert requires votes for some constant number of different messages.
    Each correct node broadcasts their votes and broadcasts zero or one approval.
    Therefore, each correct node sends at most $\mathcal{O}(n)$ messages, and all correct nodes together send at most $\mathcal{O}(n^2)$.
    Messages from byzantine nodes exceeding the maximum number can be dropped.
\end{proof}

\begin{corollary}
    If the quorums are chosen such that only a constant number of weak certificates are created, the total communication cost is $\mathcal{O}(n^2)$ messages.
\end{corollary}

\paragraph{Attribution.}
By \Cref{lem:weak-cert-safety}, if some correct node creates a weak certificate, then there is some (potentially other) correct node $v$ that observed enough votes to meet the quorum.
Thus, this node $v$ holds a cryptographic proof to attribute a set of signers to the creation of the weak certificate.

\paragraph{Transferability.}
Because all votes and approvals are signed by their creators, a correct node $v$ can always transfer their certificate (weak or strong) by sending all the votes and approvals that led $v$ to create the certificate locally to another correct node.
However, such a proof has size $\mathcal{O}(n)$.
To prevent incurring this cost, the protocol is designed in a way that these do not appear in any good-case execution.
But this transferability can be used e.g. to recover certificate propagation after a network failure.

\section{Hash-Based Signatures}
\label{sec:custom_xmss}

\begin{table}
\begin{tabular}{@{}cccccc@{}}
    \toprule
    $w$ & $T$ & Sig. & PK & Sign [$\mu s$] & Verify [$\mu s$] \\
    \midrule
    $16$  & 360     & 1{,}168~B & 24~B & 6.7  & 6.0  \\
    $32$  & 604     & 952~B     & 24~B & 11.5 & 9.4  \\
    $64$  & 1{,}008 & 784~B     & 24~B & 18.8 & 15.0 \\
    $128$ & 1{,}778 & 688~B     & 24~B & 33.7 & 26.0 \\
    $256$ & 3{,}060 & 592~B     & 24~B & 55.7 & 44.3 \\
    \bottomrule
\end{tabular}
\captionof{table}{\label{tab:wots}Comparison of various parameterizations of WOTS-TS~\cite{wots_ts}.
$w$: Winternitz space-time tradeoff.
$T$: target-sum value.
Sig./PK: signature and public-key sizes.
Sign: signing speed.
Verify: verification speed.
Measurements are from an optimized implementation of WOTS-TS on an Apple M4 Max.}
\end{table}

In this section we argue why hash-based signatures fit our setting well.
We show that integrating hash-based signatures correctly can compensate for their main downsides, which are signature size and difficulty of key management.
In return, we get the strongest post-quantum security that fits in a single MTU-sized message and has low verification cost, all while having the most conservative security assumption.
For this, we assume that the nodes already exchanged enough public keys for one-time signatures ahead of time, ideally even mapping each of them to some in-protocol notion of time (\emph{slots} or \emph{epochs} are common terms) and message type.

One-time signatures (OTSs) are a fundamental building block of hash-based signature schemes.
An OTS allows securely signing a single message under a given key, whereas reusing a key would risk signature forgery.

We consider a state-of-the-art OTS scheme based on Winternitz signatures \cite{winternitz,wotsp}, namely WOTS-TS~\cite{wots_ts}.
In \Cref{tab:wots} we look at the major parameter trade-off: signature size $w$ against signing and verification time.
Another dimension for parameterization is the target-sum value $T$, trading off signing time against verification time.
For the security level we fix 192 bits, as this is the first level that is considered more secure than the 256-bit hash collision resistance, which is an assumption many protocols already make.
A security level of 192 bits corresponds to NIST category 3 security.


As an optimization, we note that some messages in distributed protocols do not even need a full signature over a message hash.
If the signature payload is shorter than a hash itself, e.g. only a message type and an index, this message bit-string can be more efficiently signed directly.
Because for WOTS-TS the signature size is linear in the message length, hash-then-sign only makes sense once the message is longer than the hash.

\section{Performance}

\begin{table}
\begin{center}
\begin{tabular}{@{}lccc@{}}
    \toprule
    Scheme        & $n$      & Message Size & Communication \\
    \midrule
    Aggregate BLS   & 1{,}000  & 205\,B          & 285\,KB         \\
    Falcon512+LaZer & 1{,}000  & $\approx$70\,KB & $\approx$70\,MB \\
    DQS             & 1{,}000  & 624\,B          & 1{,}248\,KB     \\
    \midrule
    Aggregate BLS   & 10{,}000 & 1{,}330\,B      & 14.1\,MB    \\
    Falcon512+LaZer & 10{,}000 & $\approx$70\,KB & $\approx$700\,MB \\
    DQS             & 10{,}000 & 624\,B          & 12.5\,MB    \\
    \bottomrule
\end{tabular}
\caption{\label{tab:communication-cost}Maximum message sizes and communication cost per node. Our DQS construction is instantiated with WOTS-TS, $w=256$ (cf. \Cref{tab:wots}). We compare DQS with pre-quantum BLS signatures (BLS12-381 instantiated for small signatures) with a bitmap and with the most compact post-quantum aggregation scheme Falcon512+LaZer.
These all assume good-case execution with vote and certificate payloads of size 32 bytes.
With $n = 10{,}000$ nodes, post-quantum DQS is arguably as efficient as pre-quantum BLS!}
\end{center}
\end{table}

Compared to the use of aggregate signatures, our protocol messages are all constant size and do not require indicating the set of signers in any message, which is usually done with a bitmap.
\Cref{tab:communication-cost} shows the communication cost of one instance of our scheme compared to one instance of BLS-based vote and certificate broadcast.
It can be seen how the message size for the (larger) certificate message in the BLS case grows as the bitmap grows linearly in the number of nodes.
As a result, at $n \approx 10{,}000$ or more our scheme becomes more efficient, in terms of communication cost.


\begin{figure}
\begin{tikzpicture}
\begin{axis}[
    ybar,
    bar width=2pt,
    ymin=0, ymax=250,
    xmin=0, xmax=101,
    ylabel={Latency [ms]},
    xlabel={Nodes reached [\% of weight]},
    xtick={0, 20, 40, 60, 80, 100},
    ytick={0, 50, 100, 150, 200, 250},
    enlarge x limits = 0.0,
    ylabel style={yshift=-5pt},
    xlabel style={yshift=5pt},
    title={Latency Simulation for Global Consensus Protocol},
    width=\linewidth,
    height=0.6\linewidth,
    tick label style={font=\small},
    legend pos=north west,
    legend style={font=\tiny},
    legend image code/.code={
        \draw [#1] (0cm,-0.075cm) rectangle (0.15cm,0.15cm);
    },
    legend cell align=left,
    ymajorgrids=true,
    grid style=dashed,
    axis on top=false,
]

\addplot+[ybar, bar shift=0pt] coordinates {
    (1, 131.20854200000005)
    (2, 131.24883799999998)
    (3, 131.26513799999998)
    (4, 131.371052)
    (5, 131.37519)
    (6, 131.38091800000004)
    (7, 131.38795)
    (8, 131.39191799999998)
    (9, 131.39317400000004)
    (10, 131.39558599999998)
    (11, 131.399184)
    (12, 131.400528)
    (13, 131.40102599999997)
    (14, 131.402146)
    (15, 131.403026)
    (16, 131.404036)
    (17, 131.40423999999996)
    (18, 131.404516)
    (19, 132.49075600000003)
    (20, 133.443666)
    (21, 133.74850200000003)
    (22, 133.93785200000005)
    (23, 134.045028)
    (24, 134.100336)
    (25, 134.11203200000003)
    (26, 134.13189000000006)
    (27, 134.55988599999998)
    (28, 135.22763400000002)
    (29, 136.24213399999996)
    (30, 136.77511599999997)
    (31, 136.97915199999997)
    (32, 137.15712599999998)
    (33, 137.43356600000004)
    (34, 137.535242)
    (35, 137.60283800000002)
    (36, 137.65934600000003)
    (37, 137.72952799999996)
    (38, 137.771262)
    (39, 137.77345599999998)
    (40, 137.80532000000002)
    (41, 137.82498399999997)
    (42, 137.85723)
    (43, 137.935738)
    (44, 137.94191999999998)
    (45, 137.94313399999996)
    (46, 137.95564999999993)
    (47, 137.95980799999998)
    (48, 137.971548)
    (49, 138.29627599999995)
    (50, 139.50750599999998)
    (51, 140.31370400000003)
    (52, 140.680578)
    (53, 141.12197200000003)
    (54, 141.738554)
    (55, 142.10868199999993)
    (56, 142.69867199999996)
    (57, 142.94571600000003)
    (58, 143.55174800000006)
    (59, 144.18192200000004)
    (60, 144.57995199999996)
    (61, 144.70632799999996)
    (62, 145.00264799999997)
    (63, 145.35592400000004)
    (64, 145.82624600000003)
    (65, 146.087836)
    (66, 146.27238800000003)
    (67, 146.40625400000005)
    (68, 146.59343)
    (69, 146.93125600000002)
    (70, 147.00348799999992)
    (71, 147.08771599999994)
    (72, 147.94349599999998)
    (73, 148.23982200000003)
    (74, 149.789998)
    (75, 150.475102)
    (76, 152.868288)
    (77, 153.61120000000003)
    (78, 156.311378)
    (79, 156.60505999999998)
    (80, 157.560638)
    (81, 159.44689600000004)
    (82, 167.88523199999997)
    (83, 170.20342999999997)
    (84, 189.697608)
    (85, 194.624482)
    (86, 195.44511999999995)
    (87, 196.05792399999999)
    (88, 196.3336359999999)
    (89, 196.40518599999993)
    (90, 196.50376799999995)
    (91, 196.63853199999997)
    (92, 196.7301679999999)
    (93, 196.782368)
    (94, 196.84017400000002)
    (95, 201.12878200000003)
    (96, 207.96448799999996)
    (97, 212.52403400000003)
    (98, 212.72423999999995)
    (99, 213.00579799999997)
    (100, 222.31105000000002)
};

\addplot+[ybar, bar shift=0pt] coordinates {
    (1, 123.67673000000002)
    (2, 123.77838200000001)
    (3, 123.84550200000001)
    (4, 123.93863800000001)
    (5, 124.06213)
    (6, 124.09127400000006)
    (7, 124.13278)
    (8, 124.17476599999996)
    (9, 124.18249800000001)
    (10, 124.19391)
    (11, 124.20722400000001)
    (12, 124.21321599999999)
    (13, 124.21556199999998)
    (14, 124.22031400000003)
    (15, 124.22369600000002)
    (16, 124.22791400000001)
    (17, 124.22863600000001)
    (18, 124.22982000000002)
    (19, 125.01597400000003)
    (20, 125.99328599999997)
    (21, 126.223698)
    (22, 126.39070800000002)
    (23, 126.407826)
    (24, 126.593878)
    (25, 126.63703)
    (26, 126.67148200000001)
    (27, 126.76479000000003)
    (28, 127.47927600000001)
    (29, 127.68943200000004)
    (30, 127.85124800000003)
    (31, 128.011952)
    (32, 128.19699000000003)
    (33, 128.346492)
    (34, 128.44056799999998)
    (35, 128.518168)
    (36, 128.578582)
    (37, 128.68828600000003)
    (38, 128.7619)
    (39, 128.819236)
    (40, 128.85257399999998)
    (41, 128.87656599999997)
    (42, 128.89400200000003)
    (43, 128.91075)
    (44, 128.92793599999996)
    (45, 128.98082399999998)
    (46, 129.03245000000004)
    (47, 129.07771199999996)
    (48, 129.11240600000002)
    (49, 129.210004)
    (50, 129.31219600000003)
    (51, 129.97409599999995)
    (52, 131.293238)
    (53, 131.45137400000004)
    (54, 132.59031399999998)
    (55, 133.322226)
    (56, 133.422602)
    (57, 134.41350799999998)
    (58, 134.804758)
    (59, 135.67699999999996)
    (60, 136.801828)
    (61, 137.061732)
    (62, 137.93425200000004)
    (63, 138.139182)
    (64, 138.485894)
    (65, 140.593732)
    (66, 140.97361800000002)
    (67, 141.01430600000003)
    (68, 141.20220600000002)
    (69, 142.04892)
    (70, 142.27766999999997)
    (71, 142.28754800000002)
    (72, 142.30508400000002)
    (73, 142.350496)
    (74, 142.38170599999998)
    (75, 143.15646999999998)
    (76, 147.018508)
    (77, 147.163822)
    (78, 152.61664599999997)
    (79, 152.679764)
    (80, 152.68359400000003)
    (81, 156.28239000000002)
    (82, 164.99562600000004)
    (83, 166.928164)
    (84, 187.36460399999999)
    (85, 190.76316799999998)
    (86, 191.42953800000004)
    (87, 191.65533)
    (88, 191.80599)
    (89, 191.88584600000002)
    (90, 191.89280000000002)
    (91, 191.90736599999997)
    (92, 191.93719199999995)
    (93, 191.950006)
    (94, 191.96830800000004)
    (95, 195.91847)
    (96, 198.96474600000002)
    (97, 202.9217500000001)
    (98, 203.140442)
    (99, 203.5228260000001)
    (100, 214.670086)
};

\legend{with DQS, with BLS}

\end{axis}
\end{tikzpicture}
\caption{\label{fig:latency-histogram}
Simulation results comparing the Alpenglow consensus protocol instantiated with our scheme vs BLS aggregation, on Solana's real-world geo-distributed network distribution of about 750 nodes (in April 2026).
The median node's finalization latency increases from 130~ms to 140~ms. 
}
\Description{Latency comparison of the proposed scheme and a BLS-based approach in a global consensus setting. The plot shows the observed consensus latency as a function of the cumulative voting weight of nodes reached. The x-axis represents the fraction of total node weight included in the aggregation process, while the y-axis reports the resulting latency in milliseconds. The blue bars correspond to our scheme and the red bars to the BLS baseline. Both approaches exhibit similar latency behavior, with latency remaining relatively stable for the majority of the weight range and increasing towards the tail as nodes with higher communication delays are included. The proposed scheme achieves comparable latency to the BLS-based approach throughout the experiment, with only marginal differences across the evaluated range.}
\end{figure}
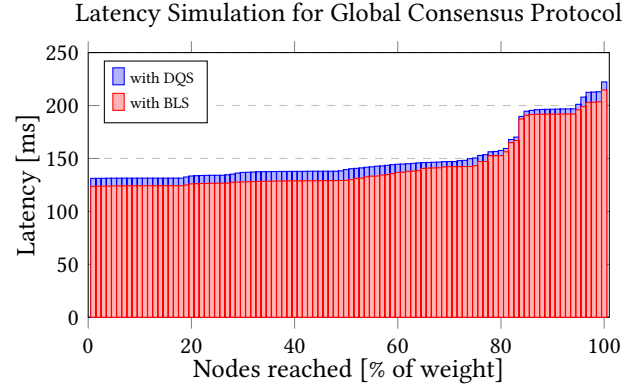

Another concern is that our scheme might increase latency in some cases, because dissemination of a strong cert to all nodes is bounded by $2 \Delta$, as seen in \Cref{lem:consistency}, whereas broadcasting an aggregate signature takes at most $\Delta$.
To prove practicality of our scheme we equipped a state-of-the-art consensus protocol with DQS in place of BLS aggregate signature broadcast.  After careful integration, the protocol does not exhibit an additional network hop on the good-case latency path, because we were able to map all events necessary for finalizing a proposal to weak certificates, with strong certificates only needed for liveness and during leader rotation.

We simulated the Alpenglow consensus~\cite{alpenglow} with the current geographic node and stake distribution of the Solana blockchain.
We assume a minimum node out-bandwidth of 1~Gbit/s, maximum node out-bandwidth of 100~Gbit/s, proportional to weight after assigning the maximum bandwidth to the node with maximum weight.
The simulation considers real internet average link latencies, per-node congestion and transmission delays.

The results of the finalization time in this simulation can be seen in \Cref{fig:latency-histogram}.
Across the entire network nodes finalize about 10~ms later, with the median node going from 130~ms to 140~ms.
As mentioned, the mapping of events to weak and strong certificates avoided extra network hops for finalization, so this difference is entirely due to additional data transmissions.


\bibliographystyle{ACM-Reference-Format}
\bibliography{references}

\end{document}